\documentclass[11pt,oneside]{amsart}
\usepackage{mathpazo}
\usepackage[utf8]{inputenc}
\usepackage{mathrsfs}
\usepackage{mathtools}
\usepackage{algorithm2e}
\usepackage{amsbsy}
\usepackage{amstext}
\usepackage{amsthm}
\usepackage{amssymb}
\usepackage{wasysym}

\makeatletter
\numberwithin{equation}{section}
\numberwithin{figure}{section}
\theoremstyle{plain}
\newtheorem{thm}{\protect\theoremname}
\theoremstyle{plain}
\newtheorem{lem}[thm]{\protect\lemmaname}

\usepackage{graphics}\usepackage{epsfig}% \usepackage{latexsym,epsfig,amstex,subfigure}
\usepackage{mathtools}
\usepackage[foot]{amsaddr}

\makeatother

\providecommand{\lemmaname}{Lemma}
\providecommand{\theoremname}{Theorem}

\begin{document}
\title{Communication-Optimal Tilings for Projective Nested Loops with Arbitrary
Bounds}
\author{Grace Dinh$^{*}$}
\thanks{$^{*}$Computer Science Division, Univ. of California, Berkeley, CA
94720 \texttt{dinh@berkeley.edu}}
\author{James Demmel$^{\dagger}$}
\thanks{$^{\dagger}$Computer Science Div. and Mathematics Dept., Univ. of
California, Berkeley, CA 94720 \texttt{demmel@berkeley.edu}}
\date{\today}

\maketitle
\global\long\def\rank{\text{{rank}}}%

\global\long\def\supp{\text{{supp}}}%

\global\long\def\st{\text{{\text{ s.t. }}}}%

\global\long\def\ov{\text{{\ensuremath{\vec{1}}}}}%

\global\long\def\sv{\text{{\ensuremath{\boldsymbol{\vec{s}}}}}}%

\begin{abstract}
Reducing communication - either between levels of a memory hierarchy
or between processors over a network - is a key component of performance
optimization (in both time and energy) for many problems, including
dense linear algebra \cite{BCD+14}, particle interactions \cite{DGKSY13},
and machine learning \cite{DD18,GAB+18}. For these problems, which
can be represented as nested-loop computations, previous tiling based
approaches \cite{CDK+13,DR16} have been used to find both lower bounds
on the communication required to execute them and optimal rearrangements,
or blockings, to attain such lower bounds. However, such general approaches
have typically assumed the problem sizes are large, an assumption
that is often not met in practice. For instance, the classical $(\text{\# arithmetic operations})/(\text{cache size})^{1/2}$
lower bound for matrix multiplication \cite{HK81,BCD+14} is not tight
for matrix-vector multiplications, which must read in at least $O(\text{\# arithmetic operations})$
words of memory; similar issues occur for almost all convolutions
in machine learning applications, which use extremely small filter
sizes (and therefore, loop bounds). 

In this paper, we provide an efficient way to both find and obtain,
via an appropriate, efficiently constructible blocking, communication
lower bounds and matching tilings which attain these lower bounds
for nested loop programs with \emph{arbitrary} loop bounds that operate
on multidimensional arrays in the \emph{projective }case, where the
array indices are subsets of the loop indices. Our approach works
on all such problems, regardless of dimensionality, size, memory access
patterns, or number of arrays, and directly applies to (among other
examples) matrix multiplication and similar dense linear algebra operations,
tensor contractions, $n$-body pairwise interactions, pointwise convolutions,
and fully connected layers.
\end{abstract}

\maketitle

\section{Introduction}

Many structured computations, including dense linear algebra, $n$-body
problems, and many machine learning kernels, can be expressed as a
collection of \emph{nested loops}, where each iteration accesses elements
from several multidimensional arrays, indexed by some function of
the current loop iteration function:
\begin{equation}
\begin{aligned} & \text{for }x_{1}\in\left[L_{1}\right],...,\text{for }x_{d}\in\left[L_{d}\right]:\\
 & \qquad\text{perform operations on }A_{1}\left[\phi_{1}\left(x_{1},...,x_{d}\right)\right],...,A_{n}\left[\phi_{n}\left(x_{1},...,x_{d}\right)\right]
\end{aligned}
\ .\label{eq:nested-loop-1}
\end{equation}

where $[L_{1}]$ represents the set $\{1,...,L_{1}\}$. For many such
problems, the time and energy costs of communication - that is, moving
data between different levels of the memory hierarchy, or between
different cores or processors - can significantly outweigh the cost
of computation in practice \cite{BCD+14}. For example, communication-optimized
implementations of matrix multiply \cite{HK81,BCD+14}, n-body problems
\cite{DGKSY13}, and convolutional neural nets \cite{GAB+18}, among
others, have significantly outperformed their non-communication-optimized
counterparts. Therefore, rearranging the order in which we perform
these operations by dividing the nested loops into subsets called
\emph{tiles} which are executed in sequence can lead to significantly
improved results in practice.

Most previous applied work, including that cited above, has been focused
on finding communication-optimal tilings and lower bounds for \emph{specific}
problems. While this is useful for commonly used kernels whose optimizations
can impact performance across a large number of applications (e.g.
matrix multiply, convolutions), it is less practicable to develop
new theory for and hand-optimize algorithms whose applications fall
into smaller niches. This has stymied research into, for instance,
unconventional neural net architectures such as capsule networks \cite{HSF18},
which require optimized kernels to test at scale but lack such kernels
due to being unproven and not widely used \cite{BI19}.

Progress has also been made \cite{CDK+13,DR16} in generalizing some
of these techniques by considering communication patterns via the
\emph{Brascamp-Lieb inequalities}, which apply to any loop nest where
the array indices are affine functions of the loop indices (i.e. the
$\phi_{i}$ above are affine). These methods provide both communication
lower bounds and constructions for tilings for such problems.

Unfortunately, the above lines of work have largely ignored situations
when certain loop bounds ($L_{i}$, above) are small. In this case,
the methods can produce weak lower bounds and infeasible tilings.
Take, for instance, the case of matrix multiplication:
\begin{align*}
 & {\rm for}\,\{x_{1},x_{2},x_{3}\}\in[L_{1}]\times[L_{2}]\times[L_{d}]\\
 & \ \ \ \ \;\;A_{1}(x_{1},x_{3})+=A_{2}(x_{1},x_{2})\times A_{3}(x_{2},x_{3})
\end{align*}

Existing combinatorial and geometric \cite{BCD+14}, techniques states
that a lower bound on the communication between a cache of size $M$
and main memory required to execute this set of instructions is
\[
\Omega\left(L_{1}L_{2}L_{3}/M^{1/2}\right)
\]
words of memory, and may be attained by rewriting the nested loops
as follows:
\begin{align*}
 & {\rm for}\,\{o_{1},o_{2},o_{3}\}\in[0..L_{1}/B_{1}-1]\times[0..L_{2}/B_{2}-1]\times[0..L_{3}/B_{3}-1]\\
 & \ \ \ \ \;\;{\rm for}\,\{i_{1},i_{2},i_{3}\}\in[B_{1}]\times[B_{2}]\times[B_{d}]\\
 & \ \ \ \ \;\;\ \ \ \ \;\;x_{1}=B_{1}o_{1}+i_{1}\\
 & \ \ \ \ \;\;\ \ \ \ \;\;x_{2}=B_{2}o_{2}+i_{2}\\
 & \ \ \ \ \;\;\ \ \ \ \;\;x_{2}=B_{2}o_{2}+i_{2}\\
 & \ \ \ \ \;\;\ \ \ \ \;\;A(x_{1},x_{3})+=A_{2}(x_{1},x_{2})\times A_{3}(x_{2},x_{3})
\end{align*}
where the \emph{tile }(the three inner loops) has dimensions $B_{1}=B_{2}=B_{3}\lessapprox\sqrt{M/3}$.

However, when $L_{1}<\sqrt{M/3}$, this tiling becomes infeasible.
Furthermore, the lower bound also ceases to be useful. For instance,
when $L_{3}=1$, corresponding to a matrix-vector multiplication,
the minimum communication needed to evaluate this multiplication is
at least $L_{1}L_{2}$, since $A_{2}$ must be read in its entirety.
However, the previous lower bound evaluates to $\Omega\left(L_{1}L_{2}/M^{1/2}\right)$,
which is clearly unachievable.

\cite{DD18} addresses this situation for convolutions, finding a
separate lower bound (and a corresponding, feasible, tiling) for the
case when the filter size is small (as they often are in most CNNs).
In this paper, we apply the techniques from \cite{DD18} to find a
\emph{general} communication lower bound and optimal tiling for \emph{arbitrary
loop bounds} in the case where the array accesses are all subsets
of the loop bounds (the so-called ``projective case'', which applies
to most dense linear algebra applications, as well as point convolutions),
and in doing so we prove that the optimal tile shape for a projective
loop nest is always a rectangle. We review the proof in the large-bound
case in Section \ref{sec:allarge}, present a stronger communication
lower bound that encompasses bounds of arbitrary size in Section \ref{sec:The-Lower-Bound},
and present a linear program that gives the actual tiling required
to achieve this lower bound (proving that it is tight) in Section
\ref{sec:Tiling-construction}. We then conclude with several examples
and a discussion in Sections \ref{sec:Examples-and-Applications}
and \ref{sec:ack}.

\section{Problem Setup, Preliminaries, Notation, and Definitions}

Define $[n]$ to be the set $\{1,2,...,n]$, and $[m,n]$ to be the
set $[m,m+1,...,n-1,n]$.

Formally, we will concern ourselves with the following $d$-level
nested-loop program, which consists of operations on the elements
of the $d_{1},...,d_{n}$-dimensional arrays $A_{1},...,A_{n}$ indexed
by affine functions $\phi_{i}:\mathbb{Z}^{d}\rightarrow\mathbb{Z}^{d_{i}}$
for $i\in[n]$:

\begin{equation}
\begin{aligned} & \text{for }x_{1}\in\left[L_{1}\right],...,\text{for }x_{d}\in\left[L_{d}\right]:\\
 & \qquad\text{perform operations on }A_{1}\left[\phi_{1}\left(x_{1},...,x_{d}\right)\right],...,A_{n}\left[\phi_{n}\left(x_{1},...,x_{d}\right)\right]
\end{aligned}
\label{eq:nested-loop}
\end{equation}

This representation includes many commonly used matrix and tensor
operations, including most linear algebra operations, tensor contractions,
and convolutional neural nets.

We will assume that each $x_{i}$ is present in the support of at
least one of the $\phi_{j}$; this assumption may be made without
loss of generality as in \cite{CDK+13}.

Let us formally model the machine as follows: suppose we have a processor
attached to a cache of size $M$, which is in turn connected to a
slow memory of unlimited size. The processor may only perform operations
on elements of the arrays present in the cache, and we wish to find
a reordering of the operations in (\ref{eq:nested-loop}) that minimizes
the amount of communication between the cache and the slow memory.

\cite{CDK+13} provides a tight lower bound for communication complexity
in this model when $L_{1},...,L_{d}$ are sufficiently large, as follows:
First, represent each operation in (\ref{eq:nested-loop}), indexed
by $x_{1},...,x_{d}$, as the point indexed by the vector $(x_{1},...,x_{d})\in\mathbb{Z}^{d}$
. As a result, the entire set of operations represented by (\ref{eq:nested-loop})
can be treated as the hyper-rectangle of $\boldsymbol{x}\in\left[L_{1}\right]\times...\times\left[L_{d}\right]$.
Furthermore, note that the element of array $A_{i}$ of memory required
for the operation indexed by $(x_{1},...,x_{d})$ is $\phi_{i}(x_{1},...,x_{d})$;
in particular, given a set $S\subset\mathbb{Z}^{d}$ of operations,
the elements of $A_{i}$ it requires are indexed by $\phi_{i}(S)$.
As a result, it suffices to find a lower bound on the number of subsets
(or an upper bound on the size of a single subset) needed to tile
the hyper-rectangle, with each one corresponding to a segment of the
program that can be executed without going back to main memory. To
satisfy this condition, we require that each subset (``tile'') $S$
satisfy the condition:
\[
\left|\phi_{i}(S)\right|\le M
\]
as we cannot use more than $M$ words memory in a computation without
going to slow memory.

Invoking the discrete Brascamp-Lieb inequality \cite{BCCT10,CDK+13},
we get that any such tile has volume at most $M^{\sum_{i\in[n]}s_{i}}$,
where the $s_{i}$ are the solutions to the linear program:
\begin{eqnarray}
\min\sum_{i\in[1..n]}s_{i}\st\label{eq:hbl}\\
\sum_{i=1}^{n}s_{i}\rank(\phi_{i}(H)) & \ge & \rank(H)\qquad\forall\text{ subgroups }H\le\mathbb{Z}^{d}\nonumber 
\end{eqnarray}
This implies that the minimum number of tiles needed to cover the
entire hyper-rectangle is at least $\prod_{i\in[1..d]}L_{i}/M^{\sum_{i\in[n]}s_{i}}$.
Since each tile corresponds to an execution of a subset of operations
without going back to slow memory, and we must complete all operations
in the 'hypercube', the total number of words transferred between
slow and fast memory must be at least
\[
\Omega\left(\frac{\prod_{i\in[d]}L_{i}}{M^{\sum_{i\in[n]}s_{i}-1}}\right)\ .
\]

An explicit construction of a tile shape that achieves this lower
bound is described in \cite{RD16}. We will review this later in the
paper for the projective case.

\section{Large Indices}

\label{sec:allarge}

In this section, we review the techniques used to construct optimal
tilings and find communication lower bounds for nested-loop programs
with large indices. Our approach in this section will be the building
block for what we do in Sections We limit our attention to cases where
the index functions, $\phi_{i}$, are projections: that is, their
output is a subset of the input. For convenience, let the indices
of the output of $\phi_{i}$ be denoted $\supp(\phi_{i})$. For instance,
if $\phi(x_{1},...,x_{5})=(x_{1},x_{4})$, then $\supp(\phi)=\{1,4\}$.

In order to find the communication lower bound, it suffices to solve
the LP (\ref{eq:hbl}). This linear program has one constraint for
each subgroup $H\le\mathbb{Z}^{d}$; since the number of such subgroups
is infinite, determining a finite closed form for inequalities using
a brute-force enumeration of all possible $H$ is impossible. Note,
however, that since $\rank(\phi_{i}(H))\le\rank(H)=d$, the number
of non-unique constraints is at most $(d+1)^{n+1}$ . In the general,
continuous case - that is, for arbitrary affine $\phi$, with $H$
ranging over subgroups of $\mathbb{R}^{d}$ rather than $\mathbb{Z}^{d}$
- an algorithm guaranteed to terminate in finite (but unbounded) time
was given by Valdimarrson \textbf{\cite{Val10}}. A separation oracle
for the resulting polytope was given in \cite{GGOW16_bl}, which immediately
implies an algorithm for enumerating the relevant constraints in double-exponential
time.

In the case where $\phi_{i}$ are projections, however, a simple,
closed-form listing of the constraints is given by Theorem 6.6 of
\cite{CDK+13}, which states that it suffices to check that the inequality
$\sum_{i=1}^{n}s_{i}\rank(\phi_{i}(H))\ge\rank(H)$ holds for all
$H$ in the set of subgroups $\{e_{1},...,e_{d}\}$, where $e_{i}$
is the subgroup comprised of all vectors with zero entries at all
indices except for $i$.

Therefore, this LP reduces to:
\begin{eqnarray}
\min\sum s_{j}\ \st\label{eq:lp_largeindex_projective}\\
1 & \le & \sum_{j\text{ s.t. }\supp(\phi_{j})\ni i}s_{j}\qquad\forall i\in[1..d]\nonumber 
\end{eqnarray}

Thinking of the $\phi_{i}$ as 0-1 vectors with $1$s in the indices
contained in its support, and letting $\boldsymbol{\vec{s}}$ denote
the vector $[s_{1},...,s_{n}]^{T}$, we can rewrite the linear program
as (omitting nonnegativity constraints) as follows: minimize $\ov^{T}\sv$
subject to:
\begin{equation}
\begin{bmatrix}\vert &  & \vert\\
\phi_{1} & \cdots & \phi_{n}\\
\vert &  & \vert
\end{bmatrix}\sv\ge\ov\ .\label{eq:hbl_largeindex_matrix_constraints}
\end{equation}

The solution to this linear program, which we denote $k_{HBL}$, immediately
gives us the communication lower bound $\prod_{i}L_{i}/M^{k_{HBL}-1}$.

Now that we have a lower bound, we would like to find an actual tiling
that attains it in order to show that it is tight. Let us ansatz\textbf{
}(following Loomis-Whitney, etc.) that the optimal tile is a hyperrectangle
of dimensions $b_{1}\times...\times b_{d}$, where the $b_{i}$ are
constants which we wish to determine. We wish to select a tile whose
volume (that is, $\prod_{i\in\{1..d\}}b_{i}$) is as large as possible,
but we are subject to memory limitations: the subsets of each array
that are used must fit in cache. Since the subsets of array $A_{i}$
required to complete the operations in this hyperrectangle are of
size $\prod_{j\in\supp(\phi_{i})}b_{j}$, we obtain the constraint
(again, ignoring constant factors) $\prod_{j\in\supp(\phi_{i})}b_{j}\le M$.
Taking logs base $M$ and letting $\lambda_{i}$ denote $\log_{M}b_{i}$,
we obtain the following linear program: maximize $\ov^{T}[\lambda_{1},...,\lambda_{d}]$
subject to:
\begin{equation}
\begin{bmatrix}- & \phi_{1} & -\\
 & \vdots\\
- & \phi_{n} & -
\end{bmatrix}\begin{bmatrix}\lambda_{1}\\
...\\
\lambda{}_{d}
\end{bmatrix}\le\ov\ .\label{eq:largealglp}
\end{equation}

Taking the dual gives us (\ref{eq:hbl_largeindex_matrix_constraints}),
which implies that this tiling obtains the lower bound.

Notice that we did not encode the constraint that $b_{i}\le L_{i}$
in this linear program. Although this does not change the result when
$L_{i}$ is assumed to be very large, this does not always hold, and
the lower bound computed by \ref{eq:hbl_largeindex_matrix_constraints}
is not always tight. In the following section, we modify this approach
to give tight lower bounds for arbitrarily-sized inputs.

\global\long\def\rank{\text{{rank}}}%

\global\long\def\supp{\text{{supp}}}%

\global\long\def\st{\text{{\text{ s.t. }}}}%

\section{The Lower Bound}

\label{sec:The-Lower-Bound}

\subsection{One Small Index}

\label{sec:HBL-Setup,-One}

We will start our approach to small loop bounds by considering the
case when all loop indices but one are assumed to be bounded by arbitrarily
large values. Our approach will be to (a) find an upper bound for
a tile restricted to single ``slice'' of the iteration space formed
by fixing the loop index with a small bound, (b) calculate an upper
bound for the entire tile by summing individual slice bounds together
over all possible values of the same index, and (c) divide the total
number of operations by the aforementioned quantity to achieve a communication
lower bound.

Let us first consider the case where a single loop bound - say, $L_{1}$,
the upper bound on $x_{1}$ - is small, and the others are large.
We may assume without loss of generality that $L_{1}\le M$; if the
opposite is true, then $L_{1}$ would be large enough for the analysis
of Section \ref{sec:allarge} to apply, as any tile whose memory footprint
is at most $M$ would fit in the $L_{1}$ dimension. Furthermore,
suppose without loss of generality that $\phi_{1},...,\phi_{p}$ (for
some integer $p$) all contain $x_{1}$ and $\phi_{p+1},...,\phi_{n}$
do not. We will now find a communication lower bound for the subset
of instructions whose $x_{1}$ index is fixed (since the loop bounds
are constant and therefore independent of $x_{1}$, the result is
the same for all possible values of $x_{1}$).

Let $\phi'_{1},...,\phi'_{p}$ be the functions with $x_{1}$ removed.
For instance, if $\phi_{1}=(x_{1},x_{2},x_{3})$, then $\phi'_{1}=(x_{2},x_{3})$.
A communication lower bound for a single ``slice'' of operations
with $x_{1}$ fixed can be found by using LP \ref{eq:lp_largeindex_projective},
with the $\phi$ replaced with $\phi'$, to compute an upper bound
for the max tile size...

\begin{eqnarray*}
\min\sum\hat{s}_{j}\ \st\\
1 & \le & \sum_{j\text{ s.t. }\supp(\phi'_{j})\ni i}\hat{s}_{j}\qquad\forall i\in[1..d]
\end{eqnarray*}

This amounts to removing the first row in the constraint matrix of
the LP (\ref{eq:hbl_largeindex_matrix_constraints}):

\[
\begin{bmatrix}\multicolumn{3}{c}{\text{[remove first row]}}\\
\vert &  & \vert\\
\phi_{1} & \cdots & \phi_{n}\\
\vert &  & \vert
\end{bmatrix}\begin{bmatrix}\hat{s}_{1}\\
\vdots\\
\hat{s}_{d}
\end{bmatrix}\ge\ov\ 
\]

To find a upper bound for the size of a tile, we sum over the upper
bounds for the size each of its slices, each of which corresponds
to a single value of $x_{1}$. Let $\phi_{1}\brokenvert_{x_{1}=k},...,\phi_{n}\brokenvert_{x_{1}=k}$
be the functions with $x_{1}$ fixed to $k$. Then, the maximum tile
size is found by maximizing the following quantity (with $V$ representing
the tile):
\begin{eqnarray}
\sum_{i\in[L_{1}]}\vert\phi_{1}\brokenvert_{x_{1}=i}(V)\vert^{\hat{s}_{1}}\dots\vert\phi_{n}\brokenvert_{x_{1}=i}(V)\vert^{\hat{s}_{n}} & = & M^{\sum_{i\in[p+1,n]}\hat{s}_{i}}\sum_{i\in[L_{1}]}\vert\phi_{1}\brokenvert_{x_{1}=i}(V)\vert^{\hat{s}_{1}}\dots\vert\phi_{p}\brokenvert_{x_{1}=i}(V)\vert^{\hat{s}_{p}}\label{eq:LagrangeSingleVarObjective}
\end{eqnarray}
subject to:
\begin{equation}
\sum_{i\in\left[L_{1}\right]}\vert\phi_{j}\brokenvert_{x_{1}=i}(V)\vert\le M\qquad\forall j\in[p]\ .\label{eq:lagrangesinglevar}
\end{equation}

We bound (\ref{eq:LagrangeSingleVarObjective}) subject to the constraints
(\ref{eq:lagrangesinglevar}), and compute the maximum tile size,
as follows:
\begin{lem}
\label{lem:1dslicingdistro}The maximum tile size for a tile $V$,
subject to the constraints that (a) that $\phi_{i}(V)\le M$ for all
$i$ and (b) the set of all distinct $x_{1}$-coordinates of elements
of $V$ is at most $L_{1}$ in cardinality (i.e. the tile fits inside
the loop bounds), is bounded above by $M^{\kappa}$, where
\[
\kappa=\max\left\{ \sum_{i=1}^{n}\hat{s}_{i}+\beta_{1}\left(1-\sum_{i=1}^{p}\hat{s}_{i}\right),\sum_{i=1}^{n}\hat{s}_{i}\right\} \ .
\]
\end{lem}

\begin{proof}
There are three cases:
\begin{enumerate}
\item If $\sum_{i\in[p]}\hat{s}_{i}<1$, the maximum of the quantity (\ref{eq:LagrangeSingleVarObjective})
is achieved when we distribute the weight across terms in the sum,
i.e. for all $j\in[1..p]$, let $\vert\phi_{j}\brokenvert_{x_{1}=i}(V)\vert=M/L_{1}$
for all $i\in[1..L_{1}]$, which leads to a tile size of $M^{\kappa}$
where 
\begin{equation}
\kappa\coloneqq\sum_{i=1}^{n}\hat{s}_{i}+\beta_{1}\left(1-\sum_{i=1}^{p}\hat{s}_{i}\right)\label{eq:sle1}
\end{equation}
 and $\beta_{1}=\log_{M}L_{1}$.
\item If $\sum_{i\in[p]}\hat{s}_{i}>1,$the maximum is achieved when we
concentrate the entire weight into one term of the sum (i.e. for all
$j\in[1..p]$, let $\vert\phi_{j}\brokenvert_{x_{1}=i'}(V)\vert=M$
for some $i'$ and let $\vert\phi_{j}\brokenvert_{x_{1}=i}(V)\vert=0$
for $i\ne i'$), which leads to a tile size of $M^{\kappa}$ where
\begin{equation}
\kappa\coloneqq\sum_{i=1}^{n}\hat{s}_{i}\ .\label{eq:sge1}
\end{equation}
\item If $\sum_{i\in[p]}\hat{s}_{i}=1$, then both (\ref{eq:sle1}) and
(\ref{eq:sge1}) are equal. Furthermore, since the only difference
between $\hat{s}$ and $s$ is that the latter must satisfy the additional
constraint $\sum_{i\in\{1..p\}}s_{i}\ge1$ in the constraint (which
is satisfied in this case by $\hat{s}$ as well), we get an upper
bound of $M^{\sum_{i=1}^{n}\hat{s}_{i}}=M^{\sum_{i=1}^{n}s_{i}}$
immediately from (\ref{eq:lp_largeindex_projective}).
\end{enumerate}
For convenience, denote $\vert\phi_{i}\brokenvert_{x_{1}=x'_{i}}(V)\vert$,
the slice of $V$ corresponding to $x_{1}'$, as $y_{i,x'_{1}}$.
We want to maximize
\[
\sum_{x_{1}=1}^{L_{1}}y_{1,x_{1}}^{\hat{s}_{1}}\dots y_{p,x_{1}}^{\hat{s}_{p}}
\]
subject to
\[
\sum_{x_{1}=1}^{L_{1}}y_{i,x_{1}}-M\le0\qquad\forall i\in[p]\ .
\]

Without loss of generality, assume all the $\hat{s}_{i}$ are positive;
if $\hat{s}_{i}=0$, then we can remove $y_{i,x_{1}}$ from both the
statement of the maximization problem (e.g. by setting it to $1$
for all $x_{i}$) and from the quantities (\ref{eq:sle1}) and (\ref{eq:sge1})\textbf{
}without affecting the rest of the proof.

Since any slack in any one of the above inequalities can be removed
by increasing one of the $y_{i,x_{i}}$, and doing so will only increase
the quantity we're trying to maximize, we can take these inequalities
to be equalities. The Lagrange multipliers for this problem are:
\begin{eqnarray*}
\mathcal{L} & = & \sum_{x_{1}=1}^{L_{1}}y_{1,x_{1}}^{\hat{s}_{1}}\dots y_{p,x_{1}}^{\hat{s}_{p}}\\
 &  & -\lambda_{1}\left(\sum_{x_{1}=1}^{L_{1}}y_{1,x_{1}}-M\right)\\
 &  & \vdots\\
 &  & -\lambda_{p}\left(\sum_{x_{1}=1}^{L_{1}}y_{p,x_{1}}-M\right)\ .
\end{eqnarray*}
Setting the gradient (with respect to both $y_{i,j}$ and $\lambda_{i}$)
to $0$, and looking at the derivative with respect to $y_{i,j}$,
we get:
\begin{eqnarray}
\hat{s}_{i}y_{1,j}^{\hat{s}_{1}}...y_{i-1,j}^{\hat{s}_{i-1}}y_{i,j}^{\hat{s}_{i}-1}y_{i+1,j}^{\hat{s}_{i+1}}...y_{p,j}^{\hat{s}_{p}} & = & \lambda_{i}\ .\label{eq:lagrangemultderiv}
\end{eqnarray}

These equations are invariant in $j$: that is, no matter which value
$j$ we fix $x_{1}$ to, the set of equations that $y_{i,j}$ must
satisfy are identical (this intuitively follows from symmetry across
the $x_{i}$).

As a result, we may assume $\lambda_{i}\ne0$; if it is in fact zero,
then the quantity we're trying to maximize would be zero, which clearly
cannot be the case since we can construct a tile containing only one
element (i.e. with our objective being $1$) that satisfies all the
constraints of the maximization problem.

In particular, $\lambda_{i}y_{i,j}/\hat{s}_{i}=y_{1,j}^{\hat{s}_{1}}...y_{p,j}^{\hat{s}_{p}}$
must remain invariant as $i$ varies (with a fixed $j$), which implies
that for any $i_{1},i_{2},j$, 
\[
\frac{\lambda_{i_{1}}y_{i_{1},j}}{\hat{s}_{i_{1}}}=\frac{\lambda_{i_{2}}y_{i_{2},j}}{\hat{s}_{i_{2}}}
\]
implying that the ratio between $y_{i,j}$ for two different values
of $i$ is independent of the $j$ (i.e. slice) we choose, remaining
fixed at
\[
\frac{y_{i_{1},j}}{y_{i_{2},j}}=\frac{\lambda_{i_{2}}\hat{s}_{i_{1}}}{\lambda_{i_{1}}\hat{s}_{i_{2}}}
\]

Therefore, the point we're trying to solve for satisfies this relationship:
\begin{equation}
y_{i,j}=\frac{\lambda_{1}\hat{s}_{i}}{\lambda_{i}\hat{s}_{1}}y_{1,j}\label{eq:yijratio}
\end{equation}
For any given $j$, one of two cases must hold: either $y_{i,j}=0$
for all $i$ (in which case the tile does not intersect at all with
the slice $x_{1}=j$) or all $y_{i,j}$ are nonzero, and we can substitute
(\ref{eq:yijratio}) into (\ref{eq:lagrangemultderiv}) to get:

\begin{eqnarray*}
\frac{\lambda_{i}}{\hat{s}_{i}} & = & y_{1,j}^{\hat{s}_{1}}...y_{i-1,j}^{\hat{s}_{i-1}}y_{i,j}^{\hat{s}_{i}-1}y_{i+1,j}^{\hat{s}_{i+1}}...y_{p,j}^{\hat{s}_{p}}\\
 & = & \frac{\prod_{k=1}^{p}y_{k,j}^{\hat{s}_{k}}}{y_{i,j}}\\
 & = & \frac{\prod_{k=1}^{p}\left(\frac{\lambda_{1}\hat{s}_{k}}{\lambda_{k}\hat{s}_{1}}y_{1,j}\right)^{\hat{s}_{k}}}{\frac{\lambda_{1}\hat{s}_{i}}{\lambda_{i}\hat{s}_{1}}y_{1,j}}\\
 & = & y_{1,j}^{-1+\sum_{k}\hat{s}_{k}}\frac{\lambda_{i}\hat{s}_{1}}{\lambda_{1}\hat{s}_{i}}\prod_{k=1}^{p}\left(\frac{\lambda_{1}\hat{s}_{k}}{\lambda_{k}\hat{s}_{1}}\right)^{\hat{s}_{k}}\\
 & = & y_{1,j}^{-1+\sum_{k}\hat{s}_{k}}\frac{\lambda_{i}\hat{s}_{1}}{\lambda_{1}\hat{s}_{i}}\left(\frac{\lambda_{1}}{\hat{s}_{1}}\right)^{\sum_{k}\hat{s}_{k}}\prod_{k=1}^{p}\left(\frac{\hat{s}_{k}}{\lambda_{k}}\right)^{\hat{s}_{k}}
\end{eqnarray*}
 Canceling $\frac{\lambda_{i}}{\hat{s}_{i}}$ from both sides, and
moving the first term in the last expression over to the left, we
get
\begin{eqnarray*}
y_{1,j}^{1-\sum_{k}\hat{s}_{k}} & = & \left(\frac{\lambda_{1}}{\hat{s}_{1}}\right)^{\sum_{k}\hat{s}_{k}-1}\prod_{k=1}^{p}\left(\frac{\hat{s}_{k}}{\lambda_{k}}\right)^{\hat{s}_{k}}\ .
\end{eqnarray*}

We may assume that $1-\sum_{k=1}^{p}\hat{s}_{k}$ is nonzero, as the
case when it is zero is covered by case (3) above. Therefore, since
the right hand side is independent of $j$, it follows that all nonzero
values of $y_{1,j}$ are equal. Since $y_{1,j}$ determines the value
of $y_{i,j}$ for all $i$ via (\ref{eq:yijratio}), it follows that
each $y_{i,j}$ must either be (a) equal to some nonzero constant
independent of $j$ or (b) be equal to zero, if and only if all $y_{i',j}$
for the same $j$ must also be zero.

Let the number of $j$ such that $y_{1,j}\ne0$ be $\vartheta$, which
must fall between $1$ and $L_{1}$ inclusive (since the number of
slices is at most equal to the loop bound corresponding to the dimension
we're summing over). Therefore, in order to satisfy (a), the remaining
$y_{i,j}$ must be equal to
\[
y_{i,j}=\frac{M}{\vartheta}\ .
\]

Substituting this into (\ref{eq:LagrangeSingleVarObjective}), we
get that the max tile size is:
\[
M^{\sum_{i\in[p+1,n]}\hat{s}_{i}}\vartheta\prod_{i=1}^{p}\left(\frac{M}{\vartheta}\right)^{\hat{s}_{i}}=M^{\sum_{i=1}^{n}\hat{s}_{i}}\vartheta^{1-\sum_{i=1}^{p}\hat{s}_{i}}
\]
so the log (base $M)$ of tile size is:
\[
\sum_{i=1}^{n}\hat{s}_{i}+\left(\log_{M}\vartheta\right)\left(1-\sum_{i=1}^{p}\hat{s}_{i}\right)\ .
\]
Therefore, either $1-\sum_{i=1}^{p}\hat{s}_{i}$ is positive, in which
case the maximum occurs when we set $\vartheta$ to $L_{1}$, giving
(recall that $\beta_{1}=\log_{M}L_{1}$):
\[
\sum_{i=1}^{n}\hat{s}_{i}+\beta_{1}\left(1-\sum_{i=1}^{p}\hat{s}_{i}\right)\ ,
\]
or $1-\sum_{i=1}^{p}\hat{s}_{i}$ is negative, in which case the maximum
occurs at $\vartheta=1$, in which case we get 
\[
\sum_{i=1}^{n}\hat{s}_{i}\ .
\]
as desired.
\end{proof}

\subsection{Multiple small bounds}

We now generalize the proof Section \ref{sec:HBL-Setup,-One}\textbf{
}to the case where multiple loop bounds are taken to be small. 

Suppose that the loops indexed by $x_{i}$ have bounds $L_{i}$. Let
$R_{j}\subseteq\{1..n\}$ denote the set of indices $i$ such that
$\supp(\phi_{i})$ contains $x_{j}$.

As before, our approach considers the communication lower bound for
a ``slice'' - that is, a subset of the iteration polytope formed
by restricting certain loop indices to fixed values - and summing
these slice lower bounds over all possible values of the fixed indices.
This time, however, each slice will be formed by simultaneously fixing
multiple indices, which we assume without loss of generality are $x_{1}$
through $x_{q}$ (the following argument holds for any $q$, and is
independent of the actual value of $q$). As was the case in the single-variable
case, an upper bound on max tile size for a single slice is given
by $M^{\sum_{j\in\{1..n\}}\hat{s}_{j}},$where $\hat{s}_{j}$ are
any nonnegative numbers that satisfy: 
\begin{eqnarray}
1 & \le & \sum_{j\text{ s.t. }\supp(\phi'_{j})\ni x_{i}}\hat{s}_{j}\label{eq:croppedmulti}
\end{eqnarray}
where $\phi'_{j}$ now corresponds to removing $x_{1},...,x_{q}$
from $\phi_{j}$ (or, alternatively, chopping off the first $q$ rows
of the HBL LP constraint matrix (\ref{eq:hbl_largeindex_matrix_constraints})).We
now develop an analog to Lemma \ref{lem:1dslicingdistro} in order
to maximize the sum of the slices over $\{x_{1},...,x_{q}\}\in\{1..L_{1}\}\times...\times\{1..L_{q}\}$.
Our main result is as follows:
\begin{thm}
\label{lem:slicing} Let $q\in[1..d]$, and $\hat{s}_{i}$ be any
nonnegative numbers satisfying 
\[
1\le\sum_{j\text{ s.t. }\supp(\phi'_{j})\ni x_{i}}\hat{s}_{j}
\]
where $\phi'_{j}$ is obtained by removing $x_{1},...,x_{q}$ from
$\phi_{j}$. Then $M^{k}$, where 
\[
k=\sum_{i=1}^{n}\hat{s}_{i}+\sum_{j\in[q]\st\sum_{i\in R_{j}}\hat{s}_{i}\le1}\left[\beta_{j}\left(1-\sum_{i\in R_{j}}\hat{s}_{i}\right)\right]
\]
represents an upper bound on the tile size.
\end{thm}

Notice that this theorem holds for all possible $q$, as well as reorderings
of the variables. As a result, this lemma in fact generates $2^{d}$
separate upper bounds for tile size (one for each subset $\mathscr{Q}$
of indices that we hold to be small). Therefore, the smallest upper
bound on tile size (which corresponds to the largest lower bound on
communication) we can achieve in this manner is $M^{\hat{k}}$ for
\[
\hat{k}=\min_{\mathscr{Q}\subseteq[d]}\sum_{i=1}^{n}\hat{s}_{\mathscr{Q},i}+\sum_{j\in\mathscr{Q}\st\sum_{i\in R_{j}}\hat{s}_{i}\le1}\left[\beta_{j}\left(1-\sum_{i\in R_{j}}\hat{s}_{\mathscr{Q},i}\right)\right]
\]
where $\hat{s}_{\mathscr{Q},i}$ is the solution to the HBL LP (\ref{eq:hbl_largeindex_matrix_constraints})
with the rows indexed by elements of $\mathscr{Q}$ removed.

\begin{proof}
By induction on $q$. The base case, for $q=1$, is simply Lemma \ref{lem:1dslicingdistro}.

Let $\hat{s}'_{i}$ be defined as $\hat{s}_{[q-1],i}$. Suppose for
induction that $M^{k}$, for 
\[
k=\sum_{i=1}^{n}\hat{s}'_{i}+\sum_{j\in[q-1]\st\sum_{i\in R_{j}}\hat{s}_{i}\le1}\left[\beta_{j}\left(1-\sum_{i\in R_{j}}\hat{s}'_{i}\right)\right]
\]
represents an upper bound on the tile size.

We start by finding an upper bound on the tile size, as before, by
summing over several ``slices'', each being defined as the subset
of the elements where $x_{1}$ through $x_{q}$ are set to fixed values.

We begin by generalizing the notion of slices to the case where multiple
indices may be small. As before, let $\phi_{i}\brokenvert_{\{\hat{x}_{1},...,\hat{x}_{q}\}}$
denote $\phi_{i}$ with $x_{j}$ fixed to $\hat{x}_{j}$ for all $j\in[q]$.
By definition, as $\phi_{i}$ only depends on indices in its support,
$\phi_{i}\brokenvert_{\{x_{1},...,x_{q}\}}$ must be identical to
$\phi_{i}\brokenvert_{\{x_{1},...,x_{q}\}\cap\supp(\phi_{i})}$.

We wish to maximize the size of the entire tile - that is, the sum
of all the sizes of the slices:

\[
\sum_{x_{1}\in[1..L_{1}],...,x_{q}\in[1..L_{q}]}\vert\phi_{1}\brokenvert_{\{x_{1},...,x_{q}\}}(V)\vert^{\hat{s}_{1}}\dots\vert\phi_{n}\brokenvert_{\{x_{1},...,x_{q}\}}(V)\vert^{\hat{s}_{n}}
\]
subject to the memory constraints
\[
\sum_{x_{k}\in[1..L_{k}]\text{ for }k\in[1..q]\cap\supp(\phi_{i})}\vert\phi_{i}\brokenvert_{\{x_{1},...,x_{q}\}}(V)\vert\le M\qquad\forall i\in\bigcup_{j\in[1..q]}R_{j}\ .
\]
As before, we will simplify our notation by defining $y_{j,\{x_{1},...,x_{q}\}}\coloneqq\vert\phi_{j}\brokenvert_{\{x_{1},...,x_{q}\}}(V)\vert$.
 Our optimization problem therefore can be rewritten as maximizing:

\begin{equation}
\sum_{x_{1}\in[1..L_{1}],...,x_{q}\in[1..L_{q}]}y_{1,\{x_{1},...,x_{q}\}}^{\hat{s}_{1}}\dots y_{n,\{x_{1},...,x_{q}\}}^{\hat{s}_{n}}\label{eq:multismallobjnointersect}
\end{equation}

subject to the constraints:
\begin{equation}
1\le\sum_{x_{k}\in[1..L_{k}]\text{ for }k\in[q]\cap\supp(\phi_{i})}y_{i,\{x_{1},...,x_{q}\}}\le M\qquad\forall i\in\bigcup_{j\in[q]}R_{j}\ .\label{multismallconstraints}
\end{equation}

The definition of $\phi_{i}\brokenvert_{\{\hat{x}_{1},...,\hat{x}_{q}\}}$
(and therefore of $y_{j,\{x_{1},...,x_{q}\}}$) requires us to further
impose an additional constraint on the solution: for all $i$, the
value of $y_{i,\{x_{1},...,x_{q}\}}$ must remain independent of indices
not in the support of $\phi_{i}$. Formally, if $x_{k}\notin\supp(\phi_{i})$,
then 
\begin{equation}
y_{i,\{x_{1},...x_{k-1},a,x_{k+1},...,x_{q}\}}=y_{i,\{x_{1},...x_{k-1,},b,x_{k+1},....,x_{q}\}}\label{eq:fixsuppconstraint}
\end{equation}
for any $a,b$. Our approach will be to find a candidate solution
which ignores this constraint, and then to show that this candidate
solution actually does satisfy (\ref{eq:fixsuppconstraint}) (i.e.
that this constraint is redundant).

Furthermore, in order to make it easier to reason about the constraints
(\ref{multismallconstraints}), we will multiply them all by the appropriate
values in order to ensure that the sum is over the same set of variables:
$x_{1}$ through $x_{q}$:

\begin{equation}
\prod_{j\in[q]\backslash\supp(\phi_{i})}L_{j}\le\sum_{x_{1}\in[1..L_{1}],...,x_{q}\in[1..L_{q}]}y_{i,\{x_{1},...,x_{q}\}}\le M\prod_{j\in[q]\backslash\supp(\phi_{i})}L_{j}\qquad\forall i\in\bigcup_{j\in[q]}R_{j}\ .\label{eq:multi_small_constraint}
\end{equation}

Since our goal is to find an upper bound on the tile size, which is
the result of this constrained maximization problem, we can remove
the lower bound constraints on $\sum_{x_{1}\in[1..L_{1}],...,x_{q}\in[1..L_{q}]}y_{i,\{x_{1},...,x_{q}\}}$
(i.e. the leftmost inequality in (\ref{eq:multi_small_constraint}))
without affecting correctness. 

The resulting problem is almost identical to that of Lemma \ref{lem:1dslicingdistro},
except with different limits (one may think of this 'flattening' the
$q$-dimensional tensor $x_{1},...,x_{q}$ into a single vector in
order to get a single sum as we did in the previous section). Recall
that none of the steps we used to compute the maximum in our proof
of Lemma \ref{lem:1dslicingdistro} actually used the value of the
right sides of the constraints, since all those constants were all
differentiated away as a constant factor when taking gradients; as
a result, the same result applies here. Specifically, the maximum
is obtained at a point specified as follows: select some subset $\mathscr{S}\subseteq\{1..L_{1}\}\times...\times\{1..L_{q}\}$
of integer tuples, which represent $x_{i}$-indices for which $y_{i,\{x_{1},..,x_{q}\}}$
will be nonzero. For each $\{x_{1},..,x_{q}\}$ in $\mathscr{S}$,
$y_{i,\{x_{1},..,x_{q}\}}$ must be equal to a constant value independent
of $\{x_{1},...,x_{q}\}$. In order to maximize (\ref{eq:multismallobjnointersect}),
we set constraint (\ref{multismallconstraints}) to obtain:

\begin{equation}
y_{i,\{x_{1},...,x_{q}\}}=\frac{M}{\vert\mathscr{S}_{i}\vert}\ \forall i\label{eq:plug_in_multi_small}
\end{equation}
where $\mathscr{S}_{i}$ is $\phi_{i}$ (restricted to $x_{1}...x_{q}$)
applied to $\mathscr{S}$.

For indices not in $\mathscr{S}$, set $y_{i,\{x_{1},...,x_{q}\}}$
to zero for all $i$. The resulting upper bound for tile size is
therefore:

\begin{eqnarray}
\sum_{x_{1},...,x_{q}\in\mathscr{S}}\prod_{i}\left(\frac{M}{\vert\mathscr{S}_{i}\vert}\right)^{\hat{s}_{i}} & = & \vert\mathscr{S}\vert\prod_{i}\left(\frac{M}{\vert\mathscr{S}_{i}\vert}\right)^{\hat{s}_{i}}\nonumber \\
 & = & \frac{\vert\mathscr{S}\vert}{\prod_{i}\vert\mathscr{S}_{i}\vert^{\hat{s}_{i}}}M^{\sum_{i}\hat{s}_{i}}\label{eq:multismallobjintermsofs}
\end{eqnarray}
where the first equality is a result of the independence of the summand
with $x$, with the number of nonzero terms in the sum being $\vert\mathscr{S}\vert$.

\textbf{Claim:} without loss of generality, we can assume that $\mathscr{S}$
is a rectangle; that is, it can be written as set $C_{1}\times\dots\times C_{q}$
for some sets $C_{i}\subseteq[L_{i}]$

\textbf{Proof of claim:} Suppose not. Then there exist points $x',x''\in\mathscr{S}$
such there exists some point $x^{*}\notin\mathscr{S}$, where each
$x_{j}^{*}$ is equal to $x_{j}'$ for all $j$ except a single value
$j^{*}$, at which it takes on the value of $x_{j}''$. To see why
this is true, take any two distinct $x',x''\in\mathscr{S}$, and repeatedly
change one component of $x'$ to match the corresponding component
of $x''$, stopping when either $x'=x''$, or $x'\notin\mathscr{S}$.
In the latter case, set $x^{*}=x'$, and let $x'$ denote its immediate
predecessor in this process. If we never end up with an $x^{*}$ for
any distinct pairs of $x'$ and $x''$ in $\mathscr{S}$, then $\mathscr{S}$
must be a rectangle.

Our goal will be to show that this configuration is suboptimal. Consider
the set of functions $\phi_{i}$ for $i\in R_{j^{*}}$, that is, the
set of functions containing $x_{j^{*}}$.

Let us consider the following categories, distinguished by how $\phi_{i}$
maps $x'$, $x''$, and $x^{*}$.
\begin{enumerate}
\item $\phi_{i}(x')=\phi_{i}(x^{*})\ne\phi_{i}(x'')$. Notice that replacing
$x''$ with $x^{*}$ in $\mathscr{S}$ either reduces $\left|\mathscr{S}_{i}\right|$
by one (if there is no other $x^{\dagger}$ such that $\phi_{i}(x^{\dagger})=\phi_{i}(x'')$)
or keeps it the same (if such an $x^{\dagger}$ exists); notice that
in the latter case, adding $x^{*}$ to $\mathscr{S}$ keeps $\left|\mathscr{S}_{i}\right|$
constant. We denote these cases (1a) and (1b) respectively.
\item $\phi_{i}(x'')=\phi_{i}(x^{*})\ne\phi_{i}(x')$. Analogously to case
(1), replacing $x'$ with $x^{*}$ either reduces $\mathscr{S}_{i}$
(case (2a)) or keeps it the same (case (2b)).
\item $\phi_{i}$ maps $x'$, $x''$, and $x^{*}$ to three distinct points.
\item $\phi_{i}(x')=\phi_{i}(x'')=\phi_{i}(x^{*})$.
\item $\phi_{i}(x')=\phi_{i}(x'')\ne\phi_{i}(x^{*})$. Notice that this
category must be empty: If $x_{j}''=x_{j}'$, then by definition this
quantity is also $x^{*}$; therefore, going to $*$ from $''$ can
only make the number of agreements better under any projection.
\end{enumerate}
In the above categories, $i$ satisfying (1) and (4) implies that
$i\notin R_{j^{*}}$, while $i$ satisfying (2) and (3) implies that
$i\in R_{j^{*}}$. We will show that $\mathscr{S}$ is suboptimal
by providing strict improvements on it.
\begin{enumerate}
\item If there are any $i$ in category (1a), we replace $x''$ with $x^{*}$
in $\mathscr{S}$, reducing $\left|\mathscr{S}_{i}\right|$. In order
to see how this change affects the values of $\mathscr{S}_{i}$ for
other $i$, we first note that for other $i\notin R_{j^{*}}$, $\phi_{i}(x')=\phi_{i}(x^{*})$,
so this change can only keep constant or decrease $\left|\mathscr{S}_{i}\right|$
for such $i$. For all $i$ in any of the other categories - (1b),
(2ab), (3), or (4) - $\left|\mathscr{S}_{i}\right|$ remains the same.
Therefore, as the value of $\left|\mathscr{S}_{i}\right|$ either
remains the same or decreases (with at least one strict decrease),
and $\left|\mathscr{S}\right|$ remains constant, we obtain a strict
increase in the value of (\ref{eq:multismallobjintermsofs}).
\item If some $i$ falling into category (3): Denote the set of $i$ such
that $\phi_{i}$ maps $x'$, $x''$, and $x^{*}$ onto different values
as $Q$. We will split into two cases, based on the values of $\sum_{i\in R_{j^{*}}}\hat{s}_{i}$: 
\begin{enumerate}
\item Suppose $\sum_{i\in R_{j^{*}}}\hat{s}_{i}\ge1$. Consider the assignment
to the $y_{i,x}$ given by $\mathscr{S}$; its objective (\ref{eq:multismallobjintermsofs})
is:
\[
\sum_{x_{1}\in[L_{1}],...,x_{j^{*}-1}\in[L_{j^{*}-1}],x_{j^{*}+1}\in[L_{j^{*}+1}],...,x_{q}\in[L_{q}]}\left(\sum_{x_{j^{*}}\in[L_{j^{*}}]}y_{1,\{x_{1},...,x_{q}\}}^{\hat{s}_{1}}\dots y_{n,\{x_{1},...,x_{q}\}}^{\hat{s}_{n}}\right)
\]
Factoring the innermost term into terms that are constant w.r.t. $x_{j^{*}}$
and those that are not, we can rewrite this as:
\[
\sum_{x_{1}\in[L_{1}],...,x_{j^{*}-1}\in[L_{j^{*}-1}],x_{j^{*}+1}\in[L_{j^{*}+1}],...,x_{q}\in[L_{q}]}\left(\prod_{i\in[n]\backslash R_{j}^{*}}y_{i,\{x_{1},...,x_{q}\}}^{\hat{s}_{i}}\sum_{x_{j^{*}}\in[L_{j^{*}}]}\prod_{i\in R_{j}^{*}}y_{i,\{x_{1},...,x_{q}\}}^{\hat{s}_{i}}\right)\ .
\]
Let us restrict our attention a single ``slice'': that is, an instance
of the term
\begin{equation}
\sum_{x_{j^{*}}\in[L_{j^{*}}]}\prod_{i\in R_{j^{*}}}y_{i,\{x_{1},...,x_{q}\}}^{\hat{s}_{i}}\label{eq:innerterm1sum}
\end{equation}
with fixed values for $x_{1}$ through $x_{q}$, excluding $x_{j^{*}}$.
By equality constraints, we get that all the nonzero values of $y_{i,\{x_{1},...,x_{q}\}}^{\hat{s}_{i}}$
must be equal to a constant independent of $x_{1},...,x_{q}$ (but
dependent on $i$). Let $m_{i}=\sum_{x_{i}^{*}\in[L_{j^{*}}]}y_{i,\{x_{1},...,x_{q}\}}$,
and $\sigma$ denote the number of $x_{j^{*}}$ such that $(x_{1},...,x_{q})$,
with all coordinates except $x_{j^{*}}$ set to our fixed values,
are in $\mathscr{S}$ (and therefore, nonzero terms in the above sum
(\ref{eq:innerterm1sum})); this restricts the nonzero values of $y_{i,\{x_{1},...,x_{q}\}}$
to $m_{i}/\sigma$. Therefore, we may rewrite the above sum as:
\begin{eqnarray*}
\sum_{x_{j^{*}}\in[L_{j^{*}}]}\prod_{i\in R_{j^{*}}}y_{i,\{x_{1},...,x_{q}\}}^{\hat{s}_{i}} & = & \sigma\prod_{i\in R_{j^{*}}}\left(\frac{m_{i}}{\sigma}\right)^{\hat{s}_{i}}\\
 & = & \sigma^{1-\sum_{i\in R_{j^{*}}}\hat{s}_{i}}\prod_{i\in R_{j^{*}}}m_{i}^{\hat{s}_{i}}
\end{eqnarray*}
As $\sum_{i\in R_{j^{*}}}\hat{s}_{i}\ge1$, the exponent of $\sigma$
in the above expression is nonpositive; therefore, this term is bounded
above by 
\[
\prod_{i\in R_{j^{*}}}m_{i}^{\hat{s}_{i}}
\]
which we get when we set $\sigma$ to $1$. As this upper bound holds
individually for each ``slice'', the value of the objective (\ref{eq:multismallobjintermsofs})
is upper bounded by setting $\sigma$ to $1$ for every slice, i.e.
adding an additional constraint forcing $L_{j^{*}}$ to 1, which is
equivalent to removing $x_{j*}$ from the problem entirely. Applying
our inductive hypothesis, we get that an upper bound is $M^{k'}$
where
\[
k'=\sum_{i=1}^{n}\hat{s}'_{i}+\sum_{j\in[1..q]\backslash\{j^{*}\}\st\sum_{i\in R_{j}}\hat{s}'_{i}\le1}\left[\beta_{j}\left(1-\sum_{i\in R_{j}}\hat{s}'_{i}\right)\right]\ .
\]
Since $\sum_{i\in R_{j^{*}}}\hat{s}_{i}\ge1$, there is no difference
between $\hat{s}_{i}'$ and $\hat{s}_{i}$ for all $i$, as the only
constraint that the former must satisfy that the latter is not required
to is $\sum_{i\in R_{j^{*}}}\hat{s}_{i}\ge1$, which holds regardless
in this case. Therefore, we can replace $\hat{s}_{i}'$ with $\hat{s}_{i}$
in order to completing the induction for the entire proof of Lemma
\ref{lem:slicing} in this particular case.
\item Now suppose $\sum_{i\in R_{j^{*}}}\hat{s}_{i}<1$. As $Q\subseteq R_{j^{*}}$,
it immediately follows that $\sum_{i\in Q}\hat{s}_{i}\le\sum_{i\in R_{j^{*}}}\hat{s}_{i}<1$.
Consider the values of $y_{i,x'}$, $y_{i,x''}$, and $y_{i,x^{*}}$;
we will show that a reassignment of these three values strictly increases
the objective.\\
Without loss of generality, we will assume $\prod_{i\notin Q}y_{i,x^{*}}^{\hat{s}_{i}}$
may be taken as nonzero. Why? Given some $i'$ is not in $Q$, then
by definition $\phi_{i'}$ must map $x',x^{*}$ to the same point,
or $x'',x^{*}$ to the same point. In the former case, we can set
$y_{i',x^{*}}=y_{i',x'}$ without violating any constraint, as we
can substitute the two terms freely in any constraint-sum involving
them; in the latter case, the same applies if we set $y_{i',x^{*}}=y_{i',x''}$.
As $x',x''\in\mathscr{S}$, both $y_{i',x'}$ and $y_{i',x''}$ must
be nonzero, so $y_{i',x^{*}}$ must be nonzero as well. As nonzero
values of $y_{i,x}$ are independent of $x$ for all $i$, we must
have 
\begin{equation}
y_{i',x^{*}}=y_{i',x''}=y_{i',x'}\label{eq:notinqequality}
\end{equation}
For all $i$, let the value of $y_{i,x'}+y_{i,x''}+y_{i,x^{*}}$ be
denoted $\mu_{i}$, and define $k',k'',k^{*}$ such that $y_{i,x'}=k'\mu_{i}$
(and likewise for $k''$, $k^{*}$); our starting configuration, with
$\mathscr{S}$ containing $x',x''$ but not $x^{*}$, is $k'=k''=1/2$,
$k^{*}=0$. So as not to break any constraints, we will require that
the value of $y_{i,x'}+y_{i,x''}+y_{i,x^{*}}$ stay constant, so we
will enforce $k'+k''+k^{*}=1$. The contribution of these three tiles
to the objective (\ref{eq:multismallobjnointersect}) is:
\begin{align*}
 & \prod_{i\in Q}y_{i,x'}^{\hat{s}_{i}}\prod_{i\notin Q}y_{i,x'}^{\hat{s}_{i}}+\prod_{i\in Q}y_{i,x''}^{\hat{s}_{i}}\prod_{i\notin Q}y_{i,x''}^{\hat{s}_{i}}+\prod_{i\in Q}y_{i,x^{*}}^{\hat{s}_{i}}\prod_{i\notin Q}y_{i,x^{*}}^{\hat{s}_{i}}\\
 & =\left(\prod_{i\in Q}y_{i,x'}^{\hat{s}_{i}}+\prod_{i\in Q}y_{i,x''}^{\hat{s}_{i}}+\prod_{i\in Q}y_{i,x^{*}}^{\hat{s}_{i}}\right)\prod_{i\notin Q}y_{i,x'}^{\hat{s}_{i}}
\end{align*}
with equality following from (\ref{eq:notinqequality}). We substitute
$y_{i,x'}=k'\mu_{i}$ and the corresponding definitions for $y_{i,x''},y_{i,x^{*}}$
to rewrite the above expression as:
\begin{align*}
 & \left(\prod_{i\in Q}\left(k'\mu_{i}\right)^{\hat{s}_{i}}+\prod_{i\in Q}\left(k''\mu_{i}\right)^{\hat{s}_{i}}+\prod_{i\in Q}\left(k^{*}\mu_{i}\right)^{\hat{s}_{i}}\right)\prod_{i\notin Q}y_{i,x'}^{\hat{s}_{i}}\\
 & =\left(k'^{\sum_{i\in Q}\hat{s}_{i}}\prod_{i\in Q}\mu_{i}^{\hat{s}_{i}}+k''^{\sum_{i\in Q}\hat{s}_{i}}\prod_{i\in Q}\mu_{i}^{\hat{s}_{i}}+k{}^{*\sum_{i\in Q}\hat{s}_{i}}\prod_{i\in Q}\mu_{i}^{\hat{s}_{i}}\right)\prod_{i\notin Q}y_{i,x'}^{\hat{s}_{i}}\\
 & =\left(k'^{\sum_{i\in Q}\hat{s}_{i}}+k''^{\sum_{i\in Q}\hat{s}_{i}}+k{}^{*\sum_{i\in Q}\hat{s}_{i}}\right)\prod_{i\notin Q}y_{i,x'}^{\hat{s}_{i}}\prod_{i\in Q}\mu_{i}^{\hat{s}_{i}}
\end{align*}
We will leave $y_{i,x'}$ constant for all $i\notin Q$, and we will
not vary $\mu_{i}$, the sum of $y_{i,x'}$, $y_{i,x''}$, and $y_{i,x^{*}}$,
so it suffices to maximize
\[
k'^{\sum_{i\in Q}\hat{s}_{i}}+k''^{\sum_{i\in Q}\hat{s}_{i}}+k{}^{*\sum_{i\in Q}\hat{s}_{i}}
\]
subject to
\[
k'+k''+k^{*}=1.
\]
As $\sum_{i\in Q}\hat{s}_{i}<1$, the solution to this maximization
problem is obtained by setting $k'=k''=k^{*}=1/3$; all other assignments
(including the current one) are suboptimal. As we do not vary any
$y_{i,x}$ for $i\notin Q$ and any $x$, this change does not affect
the constraints corresponding to any other $\phi_{i}$ than those
in $Q$, which all must be still satisfied as we do not vary $\mu_{i}$;
therefore, both these assignments satisfy the constraints (\ref{multismallconstraints}).
Therefore the current assignment under $\mathscr{S}$, with $k^{*}$
set to $0$ and $k',k''$ set to $1/2$, must be suboptimal, providing
us with our contradiction in this case. 
\end{enumerate}
\item If there exists $i$ in category (2a), but none in (1a) and (3), we
will replace $x'$ with $x^{*}$ in $\mathscr{S}$, decreasing $\left|\mathscr{S}_{i}\right|$
by one. The values of $\left|\mathscr{S}_{i}\right|$ for other $i$,
in this case, either also decrease (for other $i$s falling in case
(2a)), remain the same (for $i$s falling in cases (1b), (2b), (4)),
therefore corresponding to a strict improvement in the value of (\ref{eq:multismallobjintermsofs}).
\item If we have $i$ in categories (1b), (2b), or (4) (but none in categories
(1a), (2a), or (3), all of which were dealt with in earlier cases)
add $x^{*}$ to $\mathscr{S}$; this does not change any value of
$\mathscr{S}_{i}$, but increases $\mathscr{S}$ by $1$, leading
to a strictly improved solution.
\end{enumerate}
Each of these cases (except case (3a), which uses the inductive hypothesis),
presents a strict improvement to the value of (\ref{eq:multismallobjintermsofs}).
Therefore, $\mathscr{S}$ must not be optimum, providing a contradiction.
We can therefore conclude that optimum value of $\mathscr{S}$ must
have no triple $x',x''\in\mathscr{S}$, $x^{*}\notin\mathscr{S}$
such that $x^{*}$ agrees with $x'$ everywhere except one coordinate
where it agrees with $x''$, and therefore $\mathscr{S}$ must be
a rectangle, as desired. $\blacksquare$

Now that we've shown that $\mathscr{S}$ is a rectangle, let us assume
that its dimensions are $\rho_{1},...,\rho_{q}$. Then $\mathscr{S}$
has cardinality $\prod_{i\in[q]}\rho_{i}$, and $\mathscr{S}_{i}$
has cardinality $\prod_{j\in[q]\cap\supp(\phi_{i})}\rho_{j}$. Substituting
into (\ref{eq:multismallobjintermsofs}), we get:
\begin{eqnarray*}
\frac{\vert\mathscr{S}\vert}{\prod_{i}\vert\mathscr{S}_{i}\vert^{\hat{s}_{i}}}M^{\sum_{i}\hat{s}_{i}} & = & \frac{\prod_{i\in[q]}\rho_{i}}{\prod_{i\in[d]}\left(\prod_{j\in[q]\cap\supp(\phi_{i})}\rho_{j}\right)^{\hat{s}_{i}}}M^{\sum_{i}\hat{s}_{i}}\\
 & = & \frac{\prod_{i\in[q]}\rho_{i}}{\prod_{j\in[q]}\left(\prod_{i\in R_{j}}\rho_{j}^{\hat{s}_{i}}\right)}M^{\sum_{i}\hat{s}_{i}}\\
 & = & \prod_{j\in[q]}\rho_{j}^{1-\sum_{i\in R_{j}}\hat{s}_{i}}M^{\sum_{i}\hat{s}_{i}}
\end{eqnarray*}
Since we have full control over the value of $\rho_{i}$, we can maximize
the value of this expression by setting the $\rho_{i}$ to their maximum
possible value, $L_{i}$ if $1-\sum_{i\in R_{j}}\hat{s}_{i}\ge0$,
and to their minimum possible value, $1$, if $1-\sum_{i\in R_{j}}\hat{s}_{i}\le0$.

Therefore, the maximum value of our objective (\ref{eq:multismallobjnointersect})
is obtained at:
\[
M^{\sum_{i}\hat{s}_{i}}\prod_{j\in[q]\st\sum_{i\in R_{j}}\hat{s}_{i}\le1}L_{j}^{1-\sum_{i\in R_{j}}\hat{s}_{i}}
\]

or equivalently, $M^{k}$ where
\begin{equation}
k=\sum_{i=1}^{n}\hat{s}_{i}+\sum_{j\in[q]\st\sum_{i\in R_{j}}\hat{s}_{i}\le1}\left[\beta_{j}\left(1-\sum_{i\in R_{j}}\hat{s}_{i}\right)\right]\ .\label{eq:allqlowerbound}
\end{equation}
as desired.

Finally, we need to modify our solution to satisfy (\ref{eq:fixsuppconstraint})
with no change to the objective value.

Let $y'_{i,\{x_{1},...,x_{q}\}}$ be $\max_{x_{j}\st j\in\supp(\phi_{i})}y_{i,\{x_{1},...,x_{q}\}}$,
which takes on the value $\frac{M}{\vert\mathscr{S}_{i}\vert}$ if
there is some nonzero element of $\mathscr{S}$ that matches $(x_{1},...,x_{q})$
at the indices in the support of $\phi_{i}$, and is zero otherwise.
In order to show that this modification does not change the value
of objective (\ref{eq:multismallobjnointersect}), it suffices to
show that
\begin{equation}
y_{1,\{x_{1},...,x_{q}\}}^{\hat{s}_{1}}\dots y_{n,\{x_{1},...,x_{q}\}}^{\hat{s}_{n}}=\left(y'_{1,\{x_{1},...,x_{q}\}}\right)^{\hat{s}_{1}}\dots\left(y'_{n,\{x_{1},...,x_{q}\}}\right)^{\hat{s}_{n}}\label{eq:singleterm}
\end{equation}

Suppose $\left(x_{1},...,x_{q}\right)\in\mathscr{S}$. Both sides
are nonzero, and by equality constraint it is obvious that they must
be the same.

Suppose $\left(x_{1},...,x_{q}\right)\notin\mathscr{S}$. Clearly
the left is zero. Recall that $\mathscr{S}$ is a rectangle; that
is, it can be written as set $\{\left(x_{1},...,x_{n}\right):x_{i}\in C_{i}\forall i\}$
for some sets $C_{i}\subseteq[L_{i}]$. By definition, there must
exist some $k$ such that $x_{k}\notin C_{k}$. There must be some
some $j'$ such that $\phi_{j'}$ contains $x_{k}$; by definition,
$y'_{j',\{x_{1},...,x_{k},...,x_{j}\}}$ - and therefore, the entire
right-hand-side of (\ref{eq:singleterm}) - must be zero as well.

Furthermore, in order to show that this solution does not violate
any of the constraints, consider
\[
\sum_{x_{k}\in[1..L_{k}]\text{ for }k\in[q]\cap\supp(\phi_{i})}y'_{i,\{x_{1},...,x_{q}\}}
\]

By definition, at most $\mathscr{S}_{i}$ of these terms may be nonzero,
and each since must have value $M/\vert\mathscr{S}_{i}\vert$, this
term must be at most $M$, as desired.

Notice that this proof works if we fix any subset of $1..q$ rather
than the entire set. In other words, we can freely replace the sum
from $1$ to $q$ with a sum over any subset of $1$ to $q$ and still
get a valid upper bound (by changing the sum from $j\in[q]$ to summing
over a subset of $[q]$ in equation (\ref{eq:allqlowerbound})).
\end{proof}

\section{Tiling construction}

\label{sec:Tiling-construction}

In this section, we describe an explicit construction of a tiling
that achieves the upper bound on tile size (and therefore achieves
the lower bound on computation) from section \ref{sec:The-Lower-Bound}.

Consider the LP that gives us the tiling in this case. We start with
the linear program (\ref{eq:largealglp}) and add constraints requiring
that the blocks be no larger than the loop bounds (in log-space, $\lambda_{i}\le\beta_{i}$):

\begin{eqnarray}
\max\sum_{i\in d}\lambda_{i}\st\label{eq:algfull}\\
\sum_{i\text{ s.t. }x_{i}\in\supp(\phi_{j})}\lambda_{i} & \le & 1\qquad\forall j\in[n]\nonumber \\
\lambda_{i} & \le & \beta_{i}\qquad\forall i\in[q]\nonumber \\
\lambda_{i} & \ge & 0\qquad\forall i\in[d]\nonumber 
\end{eqnarray}

\begin{thm}
The rectangular tile with dimensions given by the solution to (\ref{eq:algfull})
has cardinality equal to one of the upper bounds for tile size from
Section \ref{sec:The-Lower-Bound} for a loop program defined by the
$\phi_{j}$; in other words, the solution to (\ref{eq:algfull}) equals
\begin{equation}
\sum_{i=1}^{n}\hat{s}_{\mathscr{Q},i}+\sum_{j\in\mathscr{Q}\st\sum_{i\in R_{j}}\hat{s}_{i}\le1}\left[\beta_{j}\left(1-\sum_{i\in R_{j}}\hat{s}_{\mathscr{Q},i}\right)\right]\label{eq:s5hblobj}
\end{equation}
for some $\mathscr{Q}\subseteq[q]$, where $\hat{s}_{\mathscr{Q},i}$
satisfies the constraints of the HBL LP (\ref{eq:hbl_largeindex_matrix_constraints})
with the rows indexed by elements of $\mathscr{Q}$ removed:\begin{equation}
\label{eq:hblconstraintst3}
\begin{bmatrix}  \multicolumn{3}{c}{\text{remove rows not in  \ensuremath{\mathscr{Q}}}}\\ \vert &  & \vert\\ \phi_{1} & \cdots & \phi_{n}\\ \vert &  & \vert \end{bmatrix}\begin{bmatrix}\hat{s}_{1}\\ \vdots\\ \hat{s}_{n} \end{bmatrix}\ge\begin{bmatrix}1\\ \vdots\\ 1 \end{bmatrix}
\end{equation}
\end{thm}

Let us write the constraints of (\ref{eq:algfull}) in the following
fashion:

\begin{equation}
\label{eq:algconstraints}\begin{tabular}{l}
$\lefteqn{\phantom{\begin{matrix} \phi_1 \\ \vdots \\\phi_n \ \end{matrix}}}$\\
$\left.\lefteqn{\phantom{\begin{matrix} b_0\\ a\\ \ddots\\ b_0\ \end{matrix}}}q\right\{$
\end{tabular}
\begin{bmatrix} & - &  & \phi_{1} &  & -\\
&  &  & \vdots\\
& - &  & \phi_{n} &  & -\\
1 & 0 & \cdots & 0 & 0 & \cdots & 0\\
0 & 1 & \cdots & 0 & 0 & \cdots & 0\\
\vdots & \vdots & \ddots & \vdots & \vdots & \ddots & \vdots\\
0 & 0 & \mathclap{\hspace{-1em}\underbrace{\makebox[7em]{$\hspace{1em}\cdots$}}_{q}} & 1 & 0 & \cdots & 0
\end{bmatrix}
%\underbrace{\vphantom{\begin{matrix}aaaa & aaaa \end{matrix}}}
\begin{bmatrix}\lambda_{1}\\ \vdots\\ \lambda_{d} \end{bmatrix}\le\begin{bmatrix}1\\ \vdots\\ 1\\ \beta_{1}\\ \vdots\\ \beta_{q} \end{bmatrix} 
\end{equation}

The dual of this linear program, with variables $\zeta_{1},...,\zeta_{q},s_{1},...,s_{n}$
is to minimize 
\begin{equation}
\sum_{i\in[q]}\beta_{i}\zeta_{i}+\sum_{j=1}^{n}s_{j}\label{eq:algdual_obj}
\end{equation}
subject to

\begin{equation}
\label{eq:algdualconstraints}
\begin{tabular}{l}
$\left.\lefteqn{\phantom{\begin{matrix} 1 \\ \vdots \\0 \end{matrix}}}q\right\{$\\
$\lefteqn{\phantom{\begin{matrix} \vdots \\ 0 \end{matrix}}}$
\end{tabular}
\begin{bmatrix}1 & \cdots & 0\\ \vdots & \ddots & \vdots & \vert &  & \vert\\ 0 & \cdots & 1 & \phi_{1} & \cdots & \phi_{n}\\ \vdots &  & \vdots & \vert &  & \vert\\ 0 & \cdots & 0 \end{bmatrix}\begin{bmatrix}\zeta_{1}\\ \vdots\\ \zeta_{q}\\ s_{1}\\ \vdots\\ s_{n} \end{bmatrix}\ge\begin{bmatrix}1\\ \vdots\\ 1 \end{bmatrix}
\end{equation}(as well as nonnegativity constraints $\zeta_{i}\ge0$ for all $i\in[q]$,
$s_{i}\ge0$ for all $i\in[n]$, which we omit from the matrix for
brevity)

We now show that the optimal value of (\ref{eq:algdual_obj}) is equivalent
to (\ref{eq:s5hblobj}) for some $\hat{s}_{i}$ satisfying (\ref{eq:hblconstraintst3}).
\begin{proof}
By induction on $q$.

For the base case, suppose $q=0$. This is just the case in Section
\ref{sec:allarge}.

Suppose for induction that the solution to 
\begin{eqnarray}
\max\sum_{i}\lambda_{i}\st\label{eq:algfull-inductive}\\
\sum_{i\text{ s.t. }x_{i}\in\supp(\phi_{j})}\lambda_{i} & \le & 1\qquad\forall j\in[n]\nonumber \\
\lambda_{i} & \le & \beta_{i}\qquad\forall i\in[q-1]\nonumber 
\end{eqnarray}
takes the form
\[
\sum_{i=1}^{n}\hat{s}_{\mathscr{Q},i}+\sum_{j\in\mathscr{Q}\st\sum_{i\in R_{j}}\hat{s}_{i}\le1}\left[\beta_{j}\left(1-\sum_{i\in R_{j}}\hat{s}_{\mathscr{Q},i}\right)\right]
\]
for some $\mathscr{Q}\subseteq[q-1]$ and $\hat{s}_{i}$ satisfying
(\ref{eq:hblconstraintst3}).

Consider the LP: minimize (\ref{eq:algdual_obj}) subject to (\ref{eq:algdualconstraints}).
Denote its solution by $\zeta'_{i}$, $s'_{i}$; we wish to discover
the minimum value of the objective (\ref{eq:algdual_obj}).

We will rewrite the LP (\ref{eq:algdualconstraints}) in such a way
that preserves the optimal value of the objective. First, we remove
one variable - say, $\zeta_{q}$ - from it. Since there is no benefit
to setting $\zeta_{q}$ any larger than necessary (it increases the
objective (\ref{eq:algdual_obj}), and does not come into play in
any other constraints) we can fix its value as necessary to ensure
that either the $q$th constraint or the nonnegativity constraint
$\zeta_{q}\ge0$ is tight. We have two cases:

Case 1: $\sum_{i\in R_{q}}s'_{i}\ge1$. In this case, the $q$th constraint
is satisfied at the optimal point regardless of the value of $\zeta'_{q}$,
so we may set $\zeta_{q}$ to $0$. Now, the objective (\ref{eq:algdual_obj})
becomes:
\[
\sum_{i=1}^{q-1}\beta_{i}\zeta_{i}+\sum_{j=1}^{n}s_{j}
\]
Since the $q$th constraint is the only one containing $\zeta_{q}$,
we can delete the $q$th column on the left block of the constraint
matrix (\ref{eq:algdualconstraints}) and remove $\zeta_{q}$ from
the LP entirely. Therefore, the resulting LP is therefore exactly
the dual of  (\ref{eq:algfull-inductive}), which, by inductive hypothesis,
has optimal objective value of the form:
\[
\sum_{i=1}^{n}\hat{s}_{\mathscr{Q},i}+\sum_{j\in\mathscr{Q}\st\sum_{i\in R_{j}}\hat{s}_{i}\le1}\left[\beta_{j}\left(1-\sum_{i\in R_{j}}\hat{s}_{\mathscr{Q},i}\right)\right]
\]
for $\mathscr{Q}\subseteq[q-1]\subset[q]$, and $\hat{s}_{\mathscr{Q},i}$
satisfying (\ref{eq:hblconstraintst3}) as desired.

Case 2: $\sum_{i\in R_{q}}s'_{i}<1$. Without loss of generality,
assume this holds for $R_{1}$ through $R_{q-1}$ as well (if not,
find $j$ such that $\sum_{x\in R_{j}}s'_{i}\ge1$, permute the LP
to swap the positions of $\zeta_{j}$ and $\zeta_{q}$, and proceed
to case 1).

Therefore, we may modify the LP by setting $\zeta_{1}$ to $1-\sum_{i\in R_{1}}s_{i}$
to keep it tight, and do the same with $\zeta_{2}$ through $\zeta_{q}$;
this does not change the optimal objective value. Removing those constraints
(since they've all been encoded into the objective), we get a new
objective to replace (\ref{eq:algdual_obj}) in our linear program:
\[
\min\sum_{i=1}^{n}s_{i}+\sum_{j=1}^{q}\left[\beta_{j}\left(1-\sum_{i\in R_{j}}s_{i}\right)\right]
\]

Furthermore, since $\sum_{i\in R_{j}}s'_{i}<1$ for all $j\in[q]$
this objective at its optimizer, $s_{1}',...,s_{q}'$, is precisely
equal to 
\[
\sum_{i=1}^{n}s'_{i}+\sum_{j\in[q]\st\sum_{i\in R_{j}}\hat{s}_{i}\le1}\left[\beta_{j}\left(1-\sum_{i\in R_{j}}s'_{i}\right)\right]
\]
which is of the same form as (\ref{eq:s5hblobj}).

Furthermore, we may remove the first $q$ constraints from (\ref{eq:algdualconstraints}),
since our choices for values of $\zeta_{1},...,\zeta_{q}$ guarantee
that they will be satisfied. The resulting constraint matrix is identical
to (\ref{eq:hblconstraintst3}).

Therefore, the tile whose dimensions are given by \ref{eq:algfull}
attains the lower bound given by Lemma \ref{lem:slicing} with $\mathscr{Q}=[q]$,
as desired.
\end{proof}

\section{Examples}

\label{sec:Examples-and-Applications}

We demonstrate several applications of our theory below.

\subsection{Matrix-Matrix and Matrix-Vector Multiplication}

We start by re-deriving the classical lower bound \cite{HK81} for
the triply-nested-loop matrix multiplication
\begin{align*}
 & {\rm for}\,\{x_{1},x_{2},x_{3}\}\in[L_{1}]\times[L_{2}]\times[L_{d}]\\
 & \ \ \ \ \;\;A_{1}(x_{1},x_{3})+=A_{2}(x_{1},x_{2})\times A_{3}(x_{2},x_{3})
\end{align*}

Our memory accesses are given by the functions:
\begin{eqnarray*}
\phi_{1}(x_{1},x_{2},x_{3}) & = & (x_{1},x_{3})\\
\phi_{2}(x_{1},x_{2},x_{3}) & = & (x_{1},x_{2})\\
\phi_{3}(x_{1},x_{2},x_{3}) & = & (x_{2},x_{3})
\end{eqnarray*}
Therefore, the HBL LP is to minimize $s_{1}+s_{2}+s_{3}$ subject
to
\begin{equation}
\begin{bmatrix}1 & 1 & 0\\
0 & 1 & 1\\
1 & 0 & 1
\end{bmatrix}\begin{bmatrix}s_{1}\\
s_{2}\\
s_{2}
\end{bmatrix}\ge\begin{bmatrix}1\\
1\\
1
\end{bmatrix}\ .\label{eq:gemm-hbl-primal}
\end{equation}
The optimal value of this LP is obtained when all the $s_{i}$ are
$1/2$, giving a tile size upper bound of $M^{1/2+1/2+1/2}=M^{3/2}$,
which provides the standard $L_{1}L_{2}L_{3}/M^{1l2}$ lower bound.

Now let us consider the case where $L_{3}$ may be small, which corresponds
to problem sizes approaching matrix-vector multiplications (which
occurs $L_{3}=1$). In this case, our tile, which has length $M^{1/2}$
in the $L_{3}$ dimension, cannot fit in our iteration space.s

We first find a lower bound. Removing the row corresponding to $x_{3}$
from (\ref{eq:gemm-hbl-primal}), we get that given any $\hat{s}_{i}$
satisfying
\begin{equation}
\begin{bmatrix}1 & 1 & 0\\
0 & 1 & 1
\end{bmatrix}\begin{bmatrix}\hat{s}_{1}\\
\hat{s}_{2}\\
\hat{s}_{3}
\end{bmatrix}\ge\begin{bmatrix}1\\
1\\
1
\end{bmatrix}\label{eq:gemm-hbl-primal-clipped}
\end{equation}
raising $M$ to the power
\[
\max\left\{ \hat{s}_{1}+\hat{s}_{2}+\hat{s}_{3},\hat{s}_{1}+\hat{s}_{2}+\hat{s}_{3}+(\log_{M}L_{3})(1-\hat{s}_{1}-\hat{s}_{3})\right\} 
\]
represents a valid upper bound on the tile size.

Since (\ref{eq:gemm-hbl-primal-clipped}) is satisfied when $\hat{s}_{2}=1$
and $\hat{s}_{1},\hat{s}_{3}=0$, this term becomes
\[
\max\left\{ 1,1+\log_{M}L_{3}\right\} 
\]
giving an upper bound of $\max\left\{ M,ML_{3}\right\} =ML_{3}$ (as
$L_{3}$ is always positive); therefore the communication lower bound
is given by
\[
\frac{L_{1}L_{2}L_{3}}{ML_{3}}M=L_{1}L_{2}\ .
\]
This is as expected, since we need to read at least $L_{1}L_{2}$,
the size of $A_{2}$, into fast memory to perform the operation. \emph{}

Now let us consider the question of finding the tile. Instantiating
LP (\ref{eq:algfull}) with the relevant values of $\phi_{1,2,3}$,
we get:
\begin{equation}
\begin{aligned}\max & \lambda_{1}+\lambda_{2}+\lambda_{3}\ \st\\
 & \lambda_{1}+\lambda_{3}\le1\\
 & \lambda_{1}+\lambda_{2}\le1\\
 & \lambda_{2}+\lambda_{3}\le1\\
 & \lambda_{3}\le\beta_{3}=\log_{M}L_{3}
\end{aligned}
\label{eq:gemmarray}
\end{equation}

There are two cases here: if $\beta_{3}\ge1$, then the last constraint
is of no relevance, so the solution becomes $3/2$, as in the case
above .

On the other hand, if $\beta_{3}\le1$, then adding the second and
fourth inequalities gives 
\begin{equation}
\lambda_{1}+\lambda_{2}+\lambda_{3}\le1+\lambda_{3}\le1+\beta_{3}\ .\label{eq:sumineqmatmul}
\end{equation}
We again split based on whether or not $\beta_{3}\ge1/2$; intuitively,
we may consider this a question of whether the $L_{3}$ is sufficiently
large (at least $\sqrt{M}$) to fit the $\sqrt{M}\times\sqrt{M}\times\sqrt{M}$
tile derived above, or whether we must modify the tile's shape to
get it to fit in the $L_{3}$ dimension.

If $\beta_{3}\ge1/2$, then the optimum for the LP without the fourth
constraint, $\lambda_{1}=\lambda_{2}=\lambda_{3}=1/2$, satisfies
the fourth constraint and is therefore optimal, leading to the same
$\sqrt{M}\times\sqrt{M}\times\sqrt{M}$ as in the ``large loop bound''
cases discussed above.

If $\beta_{3}\le1/2$, then we can set $\lambda_{3}=\beta_{3}$ to
make the fourth inequality tight, and then set $\lambda_{1}=1-\beta_{3}$
and $\lambda_{2}=\beta_{3}$ to tighten \ref{eq:sumineqmatmul} in
addition to the first inequality in the LP; as three irredundant inequalities
are tight and we only have three variables, this solution must be
optimal as well. This obtains a tile size of $M/L_{3}\times L_{3}\times L_{3}=ML_{3}$
(with a communication cost of $L_{1}L_{2}$, a quantity that is equal
to the size of $A_{2}$ and therefore must be optimal) as expected.

Alternatively, we could achieve the same tile size with a tile of
size $\sqrt{M}\times\sqrt{M}\times L_{3}$ (corresponding to $\lambda=\lambda_{2}=1/2$,
$\lambda_{3}=\beta_{3})$. In fact, the LP is optimized by any point
between the two solutions we found previously; specifically, for any
$\alpha\le1$,
\begin{eqnarray*}
\lambda_{1} & = & \alpha/2+(1-\alpha)(1-\beta_{3})\\
\lambda_{2} & = & \alpha/2+(1-\alpha)\beta_{3}\\
\lambda_{3} & = & \beta_{3}
\end{eqnarray*}
optimizes LP (\ref{eq:gemmarray}); this corresponds to a tile size
of:
\[
\frac{M^{1-\alpha/2}}{L_{3}^{1-\alpha}}\times M^{\alpha/2}L_{3}^{1-\alpha}\times L_{3}\ .
\]

When attempting to optimize this matrix multiplication on a real-world
system, we may select any tiling from the above $\alpha$-parameterized
family of optimal tilings in order to find one that runs well in practice
(e.g. inner loops being multiples of cache line lengths or vector
units).

As the communication cost's derivation is symmetrical (i.e. it continues
to be valid when we swap the subscripts) and the tile for the small-$L_{3}$
case above remains be a legal tiling if $L_{3}$ is the smallest loop
index, we obtain the following \emph{tight} lower bound for matrix
multiplication's communication cost:
\[
\max(L_{1}L_{2}L_{3}/\sqrt{M},L_{1}L_{2},L_{2}L_{3},L_{1}L_{3})
\]

\subsection{Tensor Contraction}

Let $1\le j<k-1<d$. Let us consider a tensor contraction of the form
\begin{align*}
 & {\rm for}\,\{x_{1},...,x_{d}\}\in[L_{1}]\times...\times[L_{d}]\\
 & \ \ \ \ \;\;A_{1}(x_{1},...,x_{j},x_{k},...,x_{d})+=A_{2}(i_{1},...,i_{k-1})\times A_{3}(x_{j+1},x_{d})
\end{align*}

This nested-loop model encapsulates several machine learning applications.
For instance, \emph{pointwise convolutions} - convolutions with $1\times1$
filters, often used along depth-separable convolutions \cite{HZCKWWAA17}
to mimic the effect of standard machine learning convolutions with
less memory usage, may be represented as tensor contractions:

\begin{align}
 & {\rm for}\,\{b,c,k,w,h\}=0:\{B,C,K,W,H\}-1\nonumber \\
 & \ \ \ \ \;\;Out(k,h,w,b)+=Image(w,h,c,b)\times Filter(k,c)\label{eqn_CNN}
\end{align}
The same holds for fully connected convolutional layers.

The communication lower bound for the large-loop bound case is, as
derived in \cite{CDK+13}, is $L_{1}...L_{d}/\sqrt{M}$.

We instantiate the LP \ref{eq:algfull} to get:
\[
\max\lambda_{1}+...+\lambda_{d}
\]
subject to
\begin{eqnarray*}
\lambda_{1}+...+\lambda_{j}+\lambda_{k}+....+\lambda_{d} & \le & 1\\
\lambda_{1}+...+\lambda_{k-1} & \le & 1\\
\lambda_{j+1}+...+\lambda_{d} & \le & 1\\
\lambda_{1} & \le & \beta_{1}=\log_{M}L_{1}\\
 & \vdots\\
\lambda_{d} & \le & \beta_{d}=\log_{M}L_{d}
\end{eqnarray*}
The structure of this linear program is much like that of matrix multiplication,
and it can be transformed into one identical to that for matrix multiplication.
Let $\gamma_{1}=\sum_{i\in[j]}\lambda_{i}$, $\gamma_{2}=\sum_{i\in[j+1,k-1]}\lambda_{i}$,
and $\gamma_{3}=\sum_{i\in[k,d]}\lambda_{i}$. Then we can rewrite
the linear program as maximizing $\gamma_{1}+\gamma_{2}+\gamma_{3}$
subject to:
\begin{eqnarray*}
\gamma_{1}+\gamma_{3} & \le & 1\\
\gamma_{1}+\gamma_{2} & \le & 1\\
\gamma_{2}+\gamma_{3} & \le & 1\\
\gamma_{1} & \le & \sum_{i\in[j]}\beta_{i}\\
\gamma_{2} & \le & \sum_{i\in[j+1,k-1]}\beta_{i}\\
\gamma_{3} & \le & \sum_{i\in[k,d]}\beta_{i}
\end{eqnarray*}
As this linear program is identical to that for matrix multiplication,
it immediately follows that its optimum is either $3/2$ or $1+\min\left\{ \sum_{i\in[j]}\beta_{i},\sum_{i\in[j+1,k-1]}\beta_{i},\sum_{i\in[k,d]}\beta_{i}\right\} $,
whichever is smaller for the given program.

\subsection{$n$-body Pairwise Interactions}

Suppose we have a list of $n$ objects, and each object interacts
with every other object. This comes up frequently in many scientific
computing applications (e.g. particle simulations), as well as database
joins.

The nested loops for this problem are (for some arbitrary function
$f$):

\begin{align*}
 & {\rm for}\,\{x_{1},x_{2}\}\in[L_{1}]\times[L_{2}]\\
 & \ \ \ \ \;\;A_{1}[x_{1}]=f(A_{2}[x_{1}],A_{3}[x_{3}])
\end{align*}

Instantiating \ref{eq:algfull}, we get:
\[
\begin{aligned}\max & \lambda_{1}+\lambda_{2}\ \st\\
 & \lambda_{1}\le1\\
 & \lambda_{2}\le1\\
 & \lambda_{1}\le\beta_{1}=\log_{M}L_{1}\\
 & \lambda_{2}\le\beta_{2}=\log_{M}L_{2}
\end{aligned}
\]
which gives us a maximum tile size of $\min\left\{ M^{2},L_{1}M,L_{2}M,L_{1}L_{2}\right\} $
and a maximum communication cost of $\min\left\{ L_{1}L_{2}/M,L_{2},L_{1},M\right\} $.
The last term, $M$, is a result of the assumption in our model that
each tile carries $M$ words of memory into cache. Therefore, it is
important to note that \emph{if total amount of memory required to
execute the program without going back to main memory is less than
$M$, the output of the program will still be $M$, when in the actual
cost is in fact the sum of the sizes of the matrices.}

\emph{}

\section{Discussion and Future Work}

\label{sec:Discussion-and-Future}

In this paper, we have shown a systematic, efficiently computable
way of determining optimal tilings for projective loop nests of arbitrary
size, and used it to rederive several tight lower bounds that have
hitherto largely been computed by a problem-specific approach.

Our approach reveals some structural properties of the tile as well:
All such loop nests share an optimal tile shape (rectangles). Furthermore,
as the optimal tile's dimension for \emph{any} projective loop nest
is the solution to a linearly parameterized linear program, its cardinality
for a given loop nest must be of the form $M^{f(L_{1},...,L_{d})}$
for some piecewise linear function $f$. In fact, for a given loop
nest, we may programmatically find a closed form of $f$ by feeding
LP (\ref{eq:algfull}), which calculates the dimensions of the tile,
into a multiparametric linear program solver, e.g. that of \cite{BBM03},
as in \cite{DD18}. This piecewise-linear structure has also been
previously shown to hold for convolutions \cite{DD18}, and \emph{we
conjecture that this property holds even in the general, non-projective
case as well.}

The immediate application we see for our approach is as compiler optimization
to automatically block projective nested loops. While many such common
loops have already been extensively optimized in high-performance
libraries (and some of these optimizations have been implemented in
compilers, e.g. \texttt{icc}'s --opt-matmul flag), our techniques
are fully general - applying to applications (e.g. pairwise interactions)
that do not fit this mold - and do not require programmers to have
any familiarity with specific high performance libraries, only access
to a compiler with the right optimizations.

Furthermore, as the memory model we use can be generalized to multiprocessor
machines (as in \cite{Kni15}, following the approach of \cite{ITT04}),
our work also provides evidence for the intuition that the best way
to split projective loop-nest tasks up on a multiprocessor system
is to assign each processor a rectangular subset of the iteration
space.

Our work is intended as a first step towards generally optimizing
\emph{non-projective} nested loops, such as those found in neural
nets, image processing, and other similar structured computations,
many of which lack well-studied high-performance implementations \cite{BI19}.
Algorithms to find such tilings - and the shapes thereof - are known\footnote{Such algorithms, which enumerate all the constraints of the HBL linear
program, are in general hard (double exponential in $n$ and $d$,
as of the time of publication of this paper). However, as the cost
only needs to be incurred once (e.g. during a computation of a highly
performance sensitive kernel), and as $n$ and $d$ tend to be relatively
small in practice, this is less of an impediment than it might appear
at first glance.} for problems with large indices \cite{DR16,CDK+13}; however, a general
method for addressing the small-bound case, which occurs in many applications
(including most machine learning ones, where, for instance, filter
sizes tend to vary), is still unknown, and is left to future work.

\section*{Acknowledgements}

\label{sec:ack}

We would like to thank Tarun Kathuria for helpful discussions.

This material is based upon work supported by the US Department of
Energy, Office of Science under Award Numbers 7081675 and 1772593;
Cray, under Award Number 47277; and DARPA, under Award Number FA8750-17-2-0091.

\bibliographystyle{alpha}
\bibliography{main}

\end{document}